\documentclass{acmtrans2m}

\usepackage{amssymb}
\usepackage{graphicx}

\usepackage {graphicx}
\usepackage {amsmath}
\usepackage {amssymb}

\newtheorem{theorem}{Theorem}[section]

\newtheorem{proposition}[theorem]{Proposition}

\newdef{definition}[theorem]{Definition}
\newdef{remark}[theorem]{Remark}

\usepackage{url}
\urldef{\mailsa}\path|zhengbojin@gmail.com| \urldef{\mailsb}\path{ziqinli@public2.bta.net.cn}
\urldef{\mailsc}\path|cgs@tsinghua.edu.cn| \urldef{\mailsd}\path{jimwang@tsinghua.edu.cn}

\title{A Kind of Representation of Common Knowledge and its Application in Requirements Analysis}

 \author{Bojin Zheng\\ School of Software, Tsinghua University, Beijing 100084,China \\
 South-Central University for Nationalities, Wuhan 430074, China
 \and Deyi Li and Guisheng Chen \\Institute of Chinese Electrical System Engineering, Beijing 100840, China
 \and Jianmin Wang\\ School of Software, Tsinghua University, Beijing 100084,China}

\begin{abstract}
Since the birth of software engineering, it always are recognized as one pure engineering subject,
therefore, the foundational scientific problems are not paid much attention. This paper proposes
that Requirements Analysis, the kernel process of software engineering, can be modeled based on
the concept of ``common knowledge''. Such a model would make us understand the nature of this
process. This paper utilizes the formal language as the tool to characterize the ``common
knowledge''-based Requirements Analysis model, and theoretically proves that : 1) the precondition
of success of software projects regardless of cost would be that the participants in a software
project have fully known the requirement specification, if the participants do not understand the
meaning of the other participants; 2) the precondition of success of software projects regardless
of cost would be that the union set of knowledge of basic facts of the participants in a software
project can fully cover the requirement specification, if the participants can always understand
the meaning of the other participants. These two theorems may have potential meanings to propose
new software engineering methodology.
\end{abstract}

[Methodologies]




\begin{document}

\setcounter{page}{111}

\begin{bottomstuff}
  Responding Author's address: Deyi Li, FIT 307, Tsinghua University, Beijing 100084,China
\end{bottomstuff}
\maketitle

\section{INTRODUCTION}
Since the birth of computer software engineering, it is regarded as a pure engineering subject.
Among engineering practices, the manufacture industry is recognized as a classic and mature
instance, therefore, the concepts in manufacture industry are widely imported into the field of
software engineering. The production process of computer software are argued to be similar to the
production of a cup, a car or something else. But some obvious facts tells us, the production of
software is certainly different to material production such as car, cup and so on. The most direct
evidence is that: software crises often occur. Even if two softwares are very similar, the
subsequence of software developments may totally differ. Some reports claim that high percent
software projects are over budget, behind schedule, and unreliable. Moveover, some software
productions do not gain high user satisfaction, even that are evaluated as impractical\cite{3}.

Because software engineering is hardly regarded as scientific subject, some foundational
theoretical problems are ignored deliberately or unconsciously, the research papers on the
foundational principles are little published. This paper tries to model the processes of software
development to understand the principles of software engineering, such that to provide some
inspirations to the practices.

In classical waterfall model, the full life-cycle of software development is divided as subsequent
processes, i.e., requirements analysis, design, coding, test, deployment and maintenance. Among
all these processes, requirements analysis is the most important, and compared to requirements
analysis, the other processes are quite technical, because the later processes can be more easily
finished, even some scientists claim that they can be finished automatically, if the requirement
specification is fully, clearly and correctly defined; moreover, most failures of software
projects would owe to the uncertainty of requirements. Therefore, we can focus on the nature of
requirements analysis to explorer the principles of software engineering.

Requirements analysis actually is such a process that the demander and the team of developers
achieve ``common knowledge'', that is, the demander knows what he/she wants to do, the team knows
what the demander wants to do, and the demander knows the team knows what he/she wants to do, and
so on. Based on this idea, we formally model the process of requirement analysis and theoretically
prove two theorems: 1) the precondition of success of software projects regardless of cost would
be that the participants in a software project have fully known the requirement specification, if
the participants do not understand the meaning of the other participants; 2) the precondition of
success of software projects regardless of cost would be that the union set of knowledge of the
participants in a software project can fully cover the requirement specification, if the
participants can always understand the meaning of the other participants.

\section{BASIC IDEAS}

For modeling the process of requirements analysis, we might as well discuss the common process of
requirements analysis in the software industry.

Assumed that the software development involves two persons or companies, we called them Side A and
Side B, Side A has a demand to develop a software, and Side B will take the mission. Side B should
communicate with Side A to know what Side A wants to do. During the communication, Side A and Side
B would revise the specification or make the specification more clearly. Obviously, there is a
final specification for every successful software project. So the aim of requirements analysis is
to make Side A and Side B agree on the final specification. If we see the final specification as a
``prior'', then the aim of requirements analysis is to make the both sides understand the prior and
agree on the prior. However, the subsequent problem would be:  In what circumstance we can conclude
that the project would succeed regardless of cost?  As we know, the subsequent processes after
requirements analysis would be quite technical, hence, the success of projects would seriously
depend on the success of requirements analysis. Technically, if both sides understand and agree on
the prior, the project would succeed regardless of cost.

According to the ideas above, we say that if both sides have fully agreed on the final
specification, then the project shall finally succeed regardless of cost. That is, if both sides
achieve a common knowledge on the prior, the project shall succeed regardless of cost. This
conclusion is based on the fact that if Side A knows every item in the prior, and Side B knows
also, and Side A knows Side B knows, and Side B knows Side A knows, and Side A knows Side B knows
Side A knows, and Side B knows Side A knows Side B knows,\ldots , and so on, then if enough people
force and materials are provided and enough time is given,  the software shall be developed fully
according to the final specification, and the product will be recognized by both Side A and Side
B, that is, the project succeeds.

The researches on ``common knowledge'' origin from the works of John C. Harsanyi, the 1994 Noble
Prize winner and Robert J. Aumann, the 2005 Nobel Prize winner. This concept now has been applied
into various of fields. Halpern and Moses researched the distributed systems by common knowledge
and proved that knowledge can not be gained\cite{2}. The work of K. Mani Chandy and Jayadev Misra
advanced the conclusion above, and proved that knowledge can not be gained and can not be lost in
a system which does not admit of simultaneous events, and theoretically proved that failure
detection is impossible without using time-outs\cite{1}.

In this paper, we try to use the concept of ``common knowledge''  to characterize the process of
requirements analysis, and furthermore, to discover some principles. Based on an intuition that if
two persons have common language, they would have common knowledge,vice versa, we employ formal
language to depict the new model. Of course, the intuition mentioned above can be easily extended
to more than two persons.

\section{Background of MODELING}

For simplify, we only consider two sides: the demander and the team of software developers, and
treat them as two agents.

Every element in final specification or the ``prior'' can be regarded as a point, so the prior can
be regarded as a point set, denoted as $P$. The point set of Side A is denoted as $P_1$, and that
of Side B as $P_2$. Every element in $P$ is denoted as lower character. We call the elements as
``basic facts''.

The process of requirements analysis actually is a process of communication. All the content of
communication would be a set, denoted as $\Omega$.

We denote the languages of Side A and Side B as $L_1$ and $L_2$, and the initial states are $S_1$
and $S_2$.

The predicates of ``knows'' could be encoded into formal language. We use such a code to represent
it: if Side A knows the basic fact $a$, then $a.1$ is an element of language of Side A. That is,
$a$ is the basic fact, the symbol ``.'' is used as a separator and the symbol ``1'' represents
Side A. Similarly, $a.2$ means Side B knows $a$, $a.2.1$ means Side A knows Side B knows a. If we
use a symbol ``T'' represents $a.2$, then ``T.1'' means Side A knows T, i.e., Side A knows Side B
knows $a$.

\begin{definition}[common knowledge]For a given knowledge A, if every side knows A and every side
knows every side knows A and every side knows every side knows every side knows A, etc., then we
say, A is common knowledge. Formally, if $A(.[1|2])^* \in L_1 \wedge A(.[1|2])^* \in L_2$, then we
say A is common knowledge.
\end{definition}

\begin{definition}[common knowledge of basic facts] For a given knowledge A, if $A \in P$ and every side knows A and every side
knows every side knows A and every side knows every side knows every side knows A, etc., then we
say, A is common knowledge. Formally, if $A(.[1|2])^* \in L_1 \wedge A(.[1|2])^* \in L_2$, then we
say A is common knowledge of basic facts.
\end{definition}

If without considering of the cost, then if Side A knows that Side B knows all that Side A knows
and Side A knows all that Side B knows, Side A would ``\textbf{think}'' that Side B would perform
the project well,but actually, when $P_1\subset P$ after the full communication with Side B, Side A
would make an error, because Side A does no have enough knowledge on the basic facts. Therefore, we
give such a definition on success of project (without considering of the cost).

\begin{definition}[Success of project regardless of the cost] If $L_1 = L_2$ and $P_1 \supseteq P$,
then we say the project is successful regardless of the cost.
\end{definition}

We plan to model the process of requirements analysis with two models. One model is assumed that
two agents have less intelligence, that is, they can't understand any meaning come from the other
agent, but only know the contents of messages, i.e., just considering the communication process.
The other model is assumed that two agents have full intelligence, that is, they would understand
the full meaning coming from the other agent. Obviously, human being would be stronger than that
the first model says, and weaker than that the second model says.

\section{Model for communication}
For simplicity, we assume that the final specification has three points, and denoted them as a,b
,c. That is, $P= \{a,b,c\}$.

So the language of Side A is $L_1 =<C, T, S_1, R>$, of Side B is $L_2 =<C, T, S_2, R>$, here

$T=P\cup\{1, 2\}\cup\{.\}$

$C=\{M, N, V, U, Q\}$

$R$ is the rules set.

There are two main rules:

1. Knows himself. For Side A, there exists $S_1 \rightarrow  S_1.1$; For Side B, it is $S_2
\rightarrow  S_2.2$.

2. Makes know. If Side A tells a message $M$ to Side B, then we denote it as $S_2 \xrightarrow{M}
M.1$; Correspondingly, if Side B tells a message $N$ to Side A, then we denote it as $S_1
\xrightarrow{N} N.2 $.

Therefore,\\
 $R=\{ \\
S_1 \rightarrow QV\\
S_2 \rightarrow QU\\
V \rightarrow V.1 \\
U \rightarrow U.2 \\
Q \rightarrow a|b|c \\
V \rightarrow \varepsilon\\
U \rightarrow \varepsilon\\
S_1 \xrightarrow{M} M.2V (M \in L_2)\\
S_2 \xrightarrow{N} N.1U (N \in L_1) \\
 \} $

\subsection{conditions of language equivalence}

We will illustrate the impact of rules by simple examples.

If $P_1=\{a\}$, $P_2=\{b\}$, then if rule 1 takes effects, then the language of Side A is $L_1 =
\{a\} \cup \{a(.1)^*\}$, and of Side B is $L_2 = \{b\} \cup \{b(.2)^*\}$; When Rule 2 works, if
Side A tells Side B the basic fact $a$, then the language of Side B will increase such sentences
as $a.1, a.1(.2)^*$; and when Side A tells $a.1.1$ to Side B, then such sentences as $a.1.1,
a.1.1(.2)^*$, etc.. So the language of Side A would be $\{a\} \cup \{a.1(.[1|2])^*\}\cup
\{b.2(.[1|2])^*\}$ and the language of Side B would be $\{b\} \cup \{b.2(.[1|2])^*\}\cup
\{a.1(.[1|2])^*\}$. Obviously, $L_1 \neq L_2$. That is, if $P_1 \neq P_2$, then $L_1 \neq L_2$.

And then we consider that if $P_1 = L_2$. If $P_1 = \{a\}, P_2 = \{a\}$, when Rule 1 and Rule 2
take affects, $L_1 = \{a(.[1|2])^*\}$, and $L_2 = \{a(.[1|2])^*\}$. That is, $L_1 = L_2$.

Actually, we can theoretically prove that the necessary and sufficient condition of language
equivalence is $P_1 = P_2$.

\begin{proposition}The sufficient condition of language equivalence.
if $L_1 = L_2$ then $P_1 = P_2$
\end{proposition}

\begin{proof}

According to Rule 1 and Rule 2, $L_1 = \{x.2(.[1|2])^*| x \in P_2\} \cup \{y(.[1|2])^*|y \in P_1\}
$, $L_2 = \{y.1(.[1|2])^*| y \in P_1\} \cup \{x(.[1|2])^*|x \in P_2\} $; because $L_1 = L_2$,
therefore, $\{x.2(.[1|2])^*| x \in P_2\} \subseteq {y(.[1|2])^*|y \in P_1}$, i.e., $P_1 \supseteq
P_2$,$\{y.2(.[1|2])^*| y \in P_1\} \subseteq \{x(.[1|2])^*|x \in P_2\}$,i.e.,$P_2 \supseteq P_1$.
therefore, $P_1 = P_2$.
\end{proof}

\begin{proposition}The necessary condition of language equivalence.
if $P_1 = P_2$ then $L_1 = L_2$
\end{proposition}

\begin{proof}
1. $\forall x \in P_1, x(.[1|2])^* \subseteq L_1$; \\
$L_1={x(.[1|2])^*|x \in P_1}$\\
2. $\forall x \in P_2, x(.[1|2])^* \subseteq L_2$; \\
$L_2={x(.[1|2])^*|x \in P_2}$\\
3. because $P_1 = P_2$ \\
therefore, $L_1 = L_2$.
\end{proof}

When two sides communicate, how the knowledge transfer between them? Actually, in this model, the
common knowledge of basic facts can not be obtained and can not be lost. This means, the
communication can increase the common knowledge of basic facts for both sides.
\begin{theorem} The common knowledge of basic facts can not be obtained.
\end{theorem}

\begin{theorem} The common knowledge of basic facts can not be lost.
\end{theorem}

These two theorems are easy to prove.

\subsection{Application}
When applying the model to requirement analysis, we commonly care for the common knowledge of basic
facts. According to this model, we can say, unless the knowledge set of basic facts of both sides
can cover the final specification, the project can not be guaranteed successful without considering
of the cost.

\begin{theorem}If $P_1\supseteq P and P_2\supseteq P$, then $L_1 = L_2$.
\end{theorem}
\begin{proof}
$Because P_1\supseteq P and P\supseteq P_1$ therefore, $P_1 = P$.\\
Similarly,$P_2 = P$. According to proposition 2,$L_1 = L_2$.
\end{proof}

This theorem says that, when only considering the communication and excluding the understanding
factor, only when $P_1 \supseteq P$ and $P_2 \supseteq P$, the project can succeed without
considering of the cost. But for almost projects, $P_2 \supseteq P$ are hardly possible, because
Side B hardly knows the final specification before the communication. This theorem tells the
difficulty of success of project.

\section{Model for fully understanding}

the language of Side A is $L_1 =<C,T,S_1,R>$, of Side B is $L_2 =<C,T,S_2,R>$, here

$T=P\cup\{1,2\}\cup\{.\}$

$C=\{M, N, U, V, Q\}$

There are two main rules:

1. Knows himself. For Side A, there exists $S_1 \rightarrow  S_1.1$; For Side B, it is $S_2
\rightarrow  S_2.2$.

2. Understand. If Side A tells a message $M$ to Side B, and Side B understands $M$,then we denote
it as $S_2 \xrightarrow{M} M.1$,$S_2 \xrightarrow{M} M$; Correspondingly, If Side B tells a message
$N$ to Side A, and Side A understands $N$, then we denote it as $S_1 \xrightarrow{N} N.1 $,$S_1
\xrightarrow{N} N$.

Therefore, \\
$R=\{ \\
S_1 \rightarrow QV\\
S_2 \rightarrow QU\\
V \rightarrow V.1 \\
U \rightarrow U.2 \\
Q \rightarrow a|b|c \\
V \rightarrow \varepsilon\\
U \rightarrow \varepsilon\\
S_1 \xrightarrow{M} M.2V (M \in L_2)\\
S_2 \xrightarrow{N} N.1U (N \in L_1)\\
S_1 \xrightarrow{M} MV (M \in L_2) \\
S_2 \xrightarrow{N} NU (N \in L_1)\\
 \} $

\begin{theorem}If $P_1 \cup P_2\supseteq P$, then $L_1 = L_2$.
\end{theorem}
\begin{proof}
For any given $x \in P_1$, according to rule 2,$L_2 \supseteq  \{x(.[1|2])^*\}$; \\
For any given $y \in P_1$, $L_2 \supseteq  \{y(.[1|2])^*\}$; \\
Because $P_1 \cup P_2\supseteq P$, therefore, for any given $z \in P$,$\{z(.[1|2])^*\} \subseteq
L_2$; \\
Similarly, for any given $z \in P$,$\{z(.[1|2])^*\} \subseteq
L_1$; \\
Because $P_1 \cup P_2 \subseteq P$, therefore, $L_1 = L_2$.
\end{proof}

This theorem says that, even if the participants can fully understand the other, the success of
project still are constrained to the knowledge of the participants. Lack of knowledge would lead
to the failure. When the success depends on the integrity of knowledge, to add the knowledge, for
examples, consulting to the proper experts, appealing to experienced engineers etc.. would be
helpful to the success.

This theorem can explain that some projects fail and other projects succeed, though the managers
use the same strategy, to inflate the team .

\section{CONCLUSIONS AND FUTURE WORK}
As to software engineering, all people concede that it would be somehow art, partly because the
scientific foundation is ignored. This paper tries to lay such a foundation by mathematic models.
The two proposed models tell the borderline of success of project. At the best, this paper proves
that, the software projects would fail if the knowledge of the software demander and software
provider can't cover the final software specification , that is, lack of some knowledge, even if
both sides have no any obstacle in the communication and can understand all the meanings of the
other.

In another words, we can say, outside of \emph{The Mythical Man-Month},not only the communication
processes, but the lack of knowledge would contribute more to the unavoidable failure of software
projects. This conclusion implies that software manager would pay more attention to the
completeness of knowledge.

\begin{acks} The authors gratefully thank Dr. Chunlai Zhou
for valuable discussions and helps. This work has been supported by the National Grand Fundamental
Research 973 Program of China under Grant No. 2007CB310804 and the National Natural Science
Foundation of China under Grant No. 60496323.
\end{acks}
\bibliography{NRA}

\begin{thebibliography}{}

\bibitem[\protect\citeauthoryear{Chandy and Misra}{Chandy and Misra}{1986}]{1}
{\sc Chandy, K.~M.} {\sc and} {\sc Misra, J.} 1986.
\newblock How processes learn.
\newblock {\em Distributed Computing\/}~{\em 1,\/}~1, 40--52.

\bibitem[\protect\citeauthoryear{Gibbs}{Gibbs}{1994}]{3}
{\sc Gibbs, W.~W.} 1994.
\newblock Trends in computing: Software's chronic crisis.
\newblock {\em Scientific American\/}~{\em 271,\/}~3, 72--81.

\bibitem[\protect\citeauthoryear{Halpern and Moses}{Halpern and
  Moses}{1990}]{2}
{\sc Halpern, J.~Y.} {\sc and} {\sc Moses, Y.} 1990.
\newblock Knowledge and common knowledge in a distributed environment.
\newblock {\em Journal of the ACM\/}~{\em 37,\/}~3, 549--587.

\end{thebibliography}

\end{document}